\documentclass{siamltex}

\usepackage[text={5.125in,8.25in},centering]{geometry}

\newtheorem{remark}[theorem]{Remark}
\newtheorem{assumption}[theorem]{Assumption}

\RequirePackage{amsmath, amssymb, graphicx, url,
xspace,subfigure,braket,epstopdf}
\usepackage{comment} 
\usepackage[makeroom]{cancel}

\usepackage{graphicx,subfigure}
\usepackage{epsfig} 
\usepackage{mathptmx} 
\usepackage{times} 
\usepackage{amsmath} 
\usepackage{amssymb}  
\usepackage{amsfonts}
\usepackage{euscript}
\usepackage{color}

\title{On the stability test for\\ reproducing kernel Hilbert spaces} 
 \author{%
    Mauro Bisiacco\thanks{Department of Information Engineering, University of Padova, Padova, Italy (bisiacco@dei.unipd.it)}
    \and Gianluigi Pillonetto\thanks{Department of Information Engineering, University of Padova, Padova, Italy (giapi@dei.unipd.it).
This work has been partially supported by the Italian SIR project RBSI14JYM2Learn4AP and by the PRIN project 2015 2015PJ28EP.
}
    \hfill\today}
\begin{document}

\maketitle

\begin{abstract}
Reproducing kernel Hilbert spaces (RKHSs)
are special Hilbert spaces where all the evaluation functionals are linear and bounded.
They are in one-to-one correspondence with positive definite maps called kernels.
Stable RKHSs enjoy the additional property of containing only functions 
defined over the real line
and absolutely integrable. Necessary and sufficient conditions for RKHS stability 
are known in the literature: the integral operator induced by the kernel 
must be bounded as map 
between $\mathcal{L}_{\infty}$, the space of essentially bounded
(test) functions, and $\mathcal{L}_1$, the space of absolutely integrable functions.
Considering Mercer (continuous) kernels in continuous-time and the entire 
discrete-time class, we show that the stability test 
can be reduced to the study of the kernel operator over test functions which assume  (almost everywhere)
only the values $\pm 1$. 
They represent the same functions needed to investigate stability of 
any single element in the RKHS. 
In this way, the RKHS stability test becomes
an elegant generalization of a straightforward result
concerning Bounded-Input Bounded-Output (BIBO) stability of 
a single linear time-invariant system.
\end{abstract}


\section{Introduction}
In control literature, the term system identification is used to 
indicate procedures that learn models of dynamic systems from input-output data \cite{Zadeh1956,Astrom71}.
When the system under study is linear and time-invariant,
the problem reduces to estimating particular functions known as system impulse responses.
Given any input, they permit to calculate the corresponding system output via a convolution.
For physical (causal) continous-time systems, impulse responses are real-valued functions defined
over the positive real axis $\mathbb{R}_+$.\\  
Over the years, many approaches for system identification have been suggested. 
One of the most popular relies on traditional statistics and can be called 
the \emph{classical system identification framework} as described in  \cite{Ljung:99}
and \cite{Soderstrom}. It uses finite-dimensional parameter structures and apply techniques, like maximum likelihood or prediction error methods (PEM), to estimate the parameters. The concept of discrete model order is then adopted to control their complexity.\\
\indent In recent years, the regularization techniques described in \cite{SpringerRegBook2022,SurveyKBsysid} have proved to be  
a powerful alternative to these classical procedures based on PEM.
Instead of postulating parametric structures, impulse responses are directly searched for in 
high-dimensional (possibly infinite-dimensional) spaces.
Ill-posedness is circumvented by including information on the physics
of the problem. Important regularizers proposed in the last years 
include notions of stability \cite{SS2010,SS2011,COL12a,Bottegal2017}.
In particular, Bounded-Input Bounded-Output (BIBO) stability is a fundamental concept
encountered in control theory. It ensures that a dynamic system, solicited by any bounded input, returns only
bounded outputs. The related necessary and sufficient condition is especially simple:
a linear and time-invariant system is BIBO stable if and only if the corresponding impulse response is absolutely integrable \cite{Kailath79}. 
For illustration purposes, it is now useful to reformulate this condition. 
Let $\mathcal{L}_{\infty}$ indicate the space of essentially bounded
functions of norm $\|\cdot\|_{\infty}$ while $\mathcal{L}_1$ is the space of absolutely integrable functions of norm $\| \cdot \|_1$.
Define also $\mathcal{U}_{\infty}$ as the subset of $\mathcal{L}_{\infty}$ containing test functions $u$ as follows:
\begin{equation}\label{defU} 
\mathcal{U}_{\infty} := \left\{ u \in {\mathcal L}_{\infty} \ \text{with} \ u:\mathbb{R}_+ \rightarrow \mathbb{R} \ \text{and} \ |u(t)|=1 \ \mbox{a.e.} \right\}.
\end{equation} 
Then, using $f$ to denote a system impulse response, it is immediate to prove that 
\begin{equation}\label{BIBOf}
\text{BIBO stability} \  \iff \ \sup_{u \in \mathcal{U}_{\infty}}  \int_0^{+\infty} f(t)u(t) dt =\| f\|_1< +\infty,
\end{equation}
where the equality on the r.h.s. holds since the supremum is obtained setting $u(t)$ to the sign of $f(t)$ almost everywhere.
In the regularized framework for system identification,
one key tool to include BIBO stability in the estimation process is
the use of particular spaces of functions, known as BIBO stable (or just stable) Reproducing kernel Hilbert spaces (RKHSs)  \cite{SurveyKBsysid,MathFoundStable2020,PNAS:SS2023}. Any function contained in such spaces enjoys the condition \eqref{BIBOf}.\\
\indent RKHSs theory is described in the fundamental works 
\cite{Aronszajn50,Bergman50}. Their simplest 
formulation says that they are Hilbert spaces 
where any pointwise evaluation of a function is a linear and bounded functional.
They are also in one-to-one correspondence with positive definite kernels.
In our setting, a kernel $K$ is a real-valued map over $\mathbb{R}_+ \times \mathbb{R}_+$
such that, for any $n$ and $n$-uple of nonnegative scalars $\{x_1,\ldots,x_n\}$, the $n \times n$ matrix $\mathbf{K}$,
with $K(x_i,x_j)$ as $(i,j)$ entry, is positive semidefinite.
While first applications of RKHSs regarding statistics, approximation theory and computer vision can be found in the eighties 
\cite{BerteroIEEE,Poggio90,Wahba:90}, such spaces were then introduced in machine learning by Federico Girosi  in \cite{Girosi:1998}.
Combining RKHSs and Tikhonov regularization theory \cite{Tikhonov1963,TichonovA:77} leads to   
powerful algorithms for function estimation like regularization networks/kernel-ridge regression and support vector machines
\cite{Vapnik98,Scholkopf01b,Suykens2002,Bell2004}.\\
\indent In system identification, the concept of stable RKHSs for impulse response estimation 
was introduced in
\cite{SS2010}, see also e.g.  
\cite{SS2010,COL12a,Dinuzzo12,SurveyKBsysid} for further developments. 
As already mentioned above, these spaces, induced by the so called stable kernels, are RKHSs containing
only absolutely integrable functions. 
They permit to define impulse response estimators looking for solutions that balance adherence to experimental data and a penalty term accounting for BIBO stability. This makes the search space manageable by inducing a ranking of possible solutions: among dynamic systems that describe the data in a similar way, the one that is, in some sense, more stable will be selected.\\
\indent A complete characterization
of stable RKHSs can be found in \cite{Carmeli,MathFoundStable2020}.
It makes use of the kernel operator $L_K$ 
which,
for any given kernel $K$, maps functions $u$ in $y_u$ where
$$
y_u=L_K[u] \ \iff \ y_u(t)=\int_0^{+\infty} \ K(t,\tau)u(\tau)d\tau, \ t\ge 0.
$$
Letting $\mathcal{H}$ be the RKHS induced by $K$, one then has 
\begin{equation}\label{BIBOH}
\text{BIBO stable $\mathcal{H}$} \ \iff \ y_u \in \mathcal{L}_1 \ \ \forall u \in \mathcal{L}_{\infty}. 
\end{equation}
A simple application of the closed graph theorem, along the same line e.g. of that described in 
Lemma 4.1 contained in \cite{CP18},
shows that such condition is equivalent to the continuity
of the kernel operator as a map from  $\mathcal{L}_1$ to $\mathcal{L}_{\infty}$. 
In terms of its operator norm $\|K\|_{\infty,1}$, this means that one must have $\|K\|_{\infty,1}<+\infty$. 
In light of this observation, an equivalent stability condition can be now introduced.
We use $\|\cdot\|_{\infty}$ to indicate the norm in $\mathcal{L}_{\infty}$ and denote the 
boundary of the unit ball in such space as
\begin{equation}\label{defD}
\mathcal{B}_{\infty} =\Big\{ \ u \in \mathcal{L}_{\infty}: \|u\|_{\infty} = 1 \Big\}. 
\end{equation}
Then, the following reformulation of \eqref{BIBOH} holds:
\begin{equation}\label{BIBOH2}
\text{BIBO stable $\mathcal{H}$} \ \iff \ \sup_{u \in \mathcal{B}_{\infty}} \ \big\| \int_0^{+\infty} K(\cdot,x)u(x)dx \big\|_{\mathcal{L}_1} =\|K\|_{\infty,1} < +\infty. 
\end{equation}
This equivalence is, in some sense, closer to the condition \eqref{BIBOf}. The requirement on the single system,
represented by the impulse response
$f$, has become a condition on the kernel $K$ which embeds an infinite set of impulse responses 
(all those belonging to the induced RKHS $\mathcal{H}$). However, something that \emph{disturbs} the symmetry
is the fact that the supremum is taken over  $\mathcal{B}_{\infty}$, a larger set which strictly contains 
$\mathcal{U}_{\infty}$. To fix this issue, we work in continuous-time and consider continuous kernels $K$.
More formally, the following assumption is adopted.\\

\begin{assumption}\label{MainA}
The kernel $K:\mathbb{R}_+ \times \mathbb{R}_+ \rightarrow \mathbb{R}$
is continuous. In addition, let $\mathcal{I}$ be the set containing all the $t$
such that the kernel sections $K(t,\cdot)$ are not absolutely integrable.
Then, the Lebesgue measure of  $\mathcal{I}$ is null.\\
\end{assumption}

The first point in the assumption says that we consider the fundamental
Mercer class which contains in practice all the kernels
adopted in machine learning and system identification \cite{Scholkopf01b,SpringerRegBook2022}. 
The second part regarding the Lebesgue measure of $\mathcal{I}$
guarantees that the kernel operator makes sense.
Then,  
we will show that also in the RKHS setting 
the test functions can be confined to $\mathcal{U}_{\infty}$. 
Beyond its theoretical interest, one additional advantage of this outcome is quite obvious: 
any procedure becomes simpler by restricting its domain of competency. 
An example can be found in \cite{AbsSum2020}, where this result was first proved to hold for 
block-diagonal (discrete-time) kernels. Then, it was used for studying
this class and obtaining new insights on stable kernels theory.  
Hence, it sounds likely that our result, reported formally in the following theorem,
can be important in the future to face similar problems involving explicit computation of the kernel operator norm.\\


\begin{theorem}\label{MainTh}
Let $\mathcal{H}$ be the RKHS induced by the continuous (Mercer) kernel $K: \mathbb{R}_+ \times \mathbb{R}_+ \rightarrow \mathbb{R}$.
Then, if Assumption \ref{MainA} holds, one has 
\begin{equation}\label{BIBOH3}
\text{BIBO stable $\mathcal{H}$} \ \iff \ \sup_{u \in \mathcal{U}_{\infty}} \ \big\| \int_0^{+\infty} K(\cdot,x)u(x)dx \big\|_{\mathcal{L}_1}=\|K\|_{\infty,1} < +\infty
\end{equation}
with $\mathcal{U}_{\infty}$ defined in \eqref{defU}. \\
\end{theorem}

The proof is reported below, in Section \ref{Sec2}.
It is worth also remarking that such result holds 
also considering the entire class of discrete-time kernels 
defined over $\mathbb{N} \times \mathbb{N}$. This is discussed in Section 
\ref{Sec3}.

\section{Proof of Theorem \ref{MainTh}}\label{Sec2}

As detailed below, the proof of Theorem \ref{MainTh} is not trivial also because the 
well known Bauer's theorem cannot be used in our context.
Such principle states that the maximum of a convex function
over a closed convex set is obtained at the boundary of the optimization domain,
 e.g. see \cite{RTR}[Part 4] or \cite{Kruzik2000}.
But in our case $\mathcal{U}_{\infty}$ is a proper subset of the boundary $\mathcal{B}_{\infty}$
of the unit ball.\\ 
It is natural to divide the proof in two parts,
assuming that the kernel $K$ is either stable or unstable.

\subsection{The case of stable $K$}

Let $K$ be a Mercer and stable kernel over $\mathbb{R}_+ \times \mathbb{R}_+$.
It comes that $L_K:{\mathcal L}_{\infty} \rightarrow {\mathcal L}_1$ is well defined  
with finite operator norm, i.e. 
 $$
\|K\|_{\infty,1}=\sup_{u \in \mathcal{B}_{\infty}} \ \|L_K[u]\|_1 < +\infty.
$$
Now, we want to show that 
\begin{equation}\label{StableCase}
\sup_{u\in{\mathcal B}_{\infty}} \ \|L_K[u]\|_1=\sup_{u\in{\mathcal U}_{\infty}} \ \|L_K[u]\|_1.
\end{equation}
Since ${\mathcal U}_{\infty}\subset {\mathcal B}_{\infty}$, the problem reduces to proving 
that the following statement holds true.
\begin{proposition}\label{Statement1}
 It holds that
$$
\sup_{u \in \mathcal{U}_{\infty}} \ \|L_K[u]\|_1\ge\sup_{u \in \mathcal{B}_{\infty}} \ \|L_K[u]\|_1.
$$
\end{proposition}
To prove this, some useful lemmas are first introduced.\\

\begin{lemma}\label{Lemma1}
For any couple of functions $y,w \in {\mathcal L}_1$, the map $F:{\mathbb R} \ \rightarrow \ {\mathbb R}_+$ defined by
$$
F(x):=\int_0^{+\infty} \ | y(t)+xw(t) |dt, \ x\in{\mathbb R}
$$
is convex.
\end{lemma}

\begin{proof} We need to prove that
$$
F(ax_1+(1-a)x_2) \le aF(x_1)+(1-a)F(x_2), \ \forall  a: \ 0\le a \le 1
$$
which is equivalent to showing that
\begin{eqnarray*}
&&\int_0^{+\infty} \ | y(t)+(ax_1+(1-a)x_2)w(t) |dt \le a\int_0^{+\infty} \ | y(t)+x_1w(t) |dt \\
&& \qquad \qquad \qquad +(1-a)\int_0^{+\infty} \ | y(t)+x_2w(t) |dt, \quad \forall  a: \ 0\le a \le 1.
\end{eqnarray*}
Letting 
$$z(t):=y(t)+x_2w(t) \ \ \text{and} \ \ x:=x_1-x_2,$$ 
in turn this is equivalent to
\begin{eqnarray*}
&& \int_0^{+\infty} \ | z(t)+axw(t) |dt \le a\int_0^{+\infty} \ | z(t)+xw(t) |dt \\
&& \qquad \qquad \qquad +(1-a)\int_0^{+\infty} \ | z(t) |dt, \quad \forall  a: \ 0\le a \le 1.
\end{eqnarray*}
and also to
\begin{eqnarray*}
&& \int_0^{+\infty} \ ( | z(t)+axw(t) | - | az(t)+axw(t) |) dt \\
&& \qquad \qquad \qquad  \le (1-a)\int_0^{+\infty} \ | z(t) |dt, \quad \forall  a: \ 0\le a \le 1.
\end{eqnarray*}
From $-|a-b| \le |a|-|b| \le |a-b|$, the lhs satisfies
$$
\int_0^{+\infty} ( | z(t)+axw(t) | - | az(t)+axw(t) |) dt \le
$$
$$
\le \int_0^{+\infty} \ | (z(t)+axw(t))-(az(t)+axw(t))  \ |dt=|1-a|\int_0^{+\infty} \ | z(t) |dt.
$$
Using $a\le 1$, this can be also rewritten as
$$
\int_0^{+\infty} ( | z(t)+axw(t) | - | az(t)+axw(t) |)dt \le (1-a)\int_0^{+\infty} \ | z(t) |dt
$$
and this concludes the proof.\\
\end{proof}

An immediate consequence of Lemma \ref{Lemma1} is that,
if
$$
v=u+xw \ \text{with} \  u,w \in {\mathcal L}_{\infty} \ \text{and} \  x\in{\mathbb R},
$$
then the function 
$$
G(x):=\|L_K[v]\|_1=\|L_K[u]+xL_K[w]\|_1 
$$
is convex. Therefore, using Bauer's maximum principle,
if $x \in [a,b]$ its maximum value is attained at either $x=a$ or $x=b$. 
For any $u\in{\mathcal B}_{\infty}$ we define
$$
E_u(a,b):=\left\{ t \in {\mathbb R}_+: \ a\le|u(t)|<b, \right\}, \ \ \forall a,b: \ 0 \le a < b \le 1.
$$

The next lemma then shows that the test inputs can be restricted to a proper subset of ${\mathcal B}_{\infty}$.\\

\begin{lemma}\label{Lemma2}
For any $u\in{\mathcal B}_{\infty}$, there exists $v\in{\mathcal B}_{\infty}$ such that $\|L_K[v]\|_1\ge \|L_K[u]\|_1$ and for any $t\in{\mathbb R}_+$ it holds that $|v(t)|\ge \frac{1}{2}$.
\end{lemma}
\begin{proof}
Let  $\psi_E(x)$ denote the indicator function of the set $E$ and consider
$$
w(t,x):=u(t)+(x-1)\psi_{E_{u}(\frac{1}{4},\frac{1}{2})}(t)u(t). 
$$
One can easily see that $w(t,x)=u(t)$ outside $E_{u}(\frac{1}{4},\frac{1}{2})$, while $w(t,x)=xu(t)$ in $E_{u}(\frac{1}{4},\frac{1}{2})$. Therefore $w(\cdot,x)\in{\mathcal B}_{\infty}$ at least in the interval $x \in [-2,2]$. So, by the convexity property discussed before, one has
$$
\|L_K[w(\cdot,x)]\|_1\ge \|L_K[u]\|_1
$$
for $x=x_0$, where either $x_0=-2$ or $x_0=2$. Now, define $z(t):=w(t,x_0)$ and note that 
$$
z \in {\mathcal B}_{\infty}, \ \|L_K[z]\|_1\ge \|L_K[u]\|_1, \ \ E_{z}(\frac{1}{4},\frac{1}{2})=\emptyset
$$
since all the values in the interval $[\frac{1}{4},\frac{1}{2})$ have been moved to the interval  $[\frac{1}{2},1)$. 
Define also
$$
y(t,x):=z(t)+x\psi_{E_{z}(0,\frac{1}{4})}(t)
$$
and note that such function belongs to ${\mathcal B}_{\infty}$ (at least) if $x \in [-\frac{3}{4},\frac{3}{4}]$.
Choosing again $x_0$ as the maximizing point, it comes that $v(t):=y(t,x_0)$ satisfies
$$
v \in {\mathcal B}_{\infty}, \ \|L_K[v]\|_1 \ge \|L_K[z]\|_1 \ge \|L_K[u]\|_1, \ |v(t)| \ge \frac{1}{2}, \ \forall t\ge 0
$$
and this concludes the proof. \\
\end{proof}

\begin{remark}\label{Remark2}
If the set $E_u(\frac{1}{4},\frac{1}{2})$ has null Lebesgue measure, it would be pointless to replace $u$ with $z$ since only what happens a.e. is relevant. Such situation would e.g. arise if $u=0$ a.e. or $u \in {\mathcal U}_{\infty}$ (which would lead to an empty set). 
In these two limit cases, one could directly assume either $v=1$ a.e. or $v=u$ everywhere.\\
\end{remark}

The next lemma allows us to conclude that the set of test inputs can be further restricted to ${\mathcal U}_{\infty}$.\\

\begin{lemma}\label{Lemma3} For any $u\in{\mathcal B}_{\infty}$ such that $|u(t)|\ge\frac{1}{2}$ for any $t \in {\mathbb R}_+$, and for any $\epsilon>0$, there exists $v\in{\mathcal U}_{\infty}$ such that $\|L_K[v]\|_1\ge \|L_K[u]\|_1-\epsilon$.
\end{lemma}
\begin{proof} First, we introduce two key sequences of functions denoted by $w_n(t,x)$ 
and $v_n \in {\mathcal B}_{\infty}, \ n\in{\mathbb N}$. As clear in what follows, 
they are built recursively, one after the other. 
First, the two sequences satisfy
\begin{eqnarray*}
v_1(t)&:=&u(t) \\ 
w_{n+1}(t,x)&=&v_n(t)+(x-1)\psi_{E_{v_n}(\frac{2^n-1}{2^n},\frac{2(2^n-1)}{2^{n+1}-1})}(t)v_n(t), \ \forall n \ge 1.
\end{eqnarray*}
Just using this partial information on their nature, one can already 
see that $v_n(\cdot)\in{\mathcal B}_{\infty}$ implies $$w_{n+1}(\cdot,x) \in {\mathcal B}_{\infty}$$ at least in the interval 
$$
x \in 
\Big[\frac{1-2^{n+1}}{2(2^n-1)} ,\frac{2^{n+1}-1}{2(2^n-1)}\Big].
$$ 
So, again by the convexity property discussed before, 
it holds that
\begin{equation}\label{Constraint}
\|L_K[w_{n+1}(\cdot,x)]\|_1\ge \|L_K[v_n(\cdot)]\|_1 \ \ \text{for $x=x_n$}
\end{equation}
where 
$$
x_n=\frac{1-2^{n+1}}{2(2^n-1)} \ \ \text{or} \ \  x_n=\frac{2^{n+1}-1}{2(2^n-1)}.
$$ 
Now, we complete the definition of the sequences by
letting 
$$
v_{n+1}(t):=w_{n+1}(t,x_n)
$$ where $x_n$ is recursively defined
as that value satisfying \eqref{Constraint}.
One can easily see that $$E_{v_{n+1}}(0,\frac{2^{n+1}-1}{2^{n+1}})=\emptyset.$$
Furthermore, introducing 
$$
u_n(t):=\text{sign}[v_n(t)], \ \forall n\ge 1 \ \text{and} \ t\ge 0
$$
one can also see that 
$$
\|u_n-v_n\|_{\infty}\le \frac{1}{2^{n+1}}, \ u_n \in {\mathcal U}_{\infty}.
$$
Therefore, by the stability assumption (continuity of the operator $L_K$), it follows that
$$
\|L_{K}[u_n]-L_{K}[v_n]\|_1\le \frac{\|K\|_{\infty,1}}{2^{n+1}}=:a_n \ \Rightarrow \ \|L_{K}[u_n]\|_1\ge \|L_{K}[v_n]\|_1-a_n
$$
with $a_n\ge 0$ monotone non-decreasing and converging to zero. This, together with the monotone non-decreasing property of $\|L_{K}[v_n]\|_1$, implies $\|L_{K}[u_n]\|_1\ge\|L_{K}[u]\|_1-a_n$. We can now choose 
any $n$ such that $a_n \le \epsilon$, obtaining that $v:=u_n$ satisfies the statement present in this lemma.\\
\end{proof}

\begin{remark} \label{Remark3} The same caveats contained in Remark \ref{Remark2} hold here. 
Specifically, if $u\in{\mathcal U}_{\infty}$ one could just assume $v=u$.\\ 
\end{remark}

Now we are ready to prove Statement \ref{Statement1}. Let $u_n \in {\mathcal B}_{\infty}$ be any sequence such that $\|L_{K}[u_n]\|_1$ tends to $\|K\|_{\infty,1}$ and define
$\epsilon_n:=\frac{1}{n}$. We also build $v_n\in {\mathcal B}_{\infty}$ in accordance with Lemma \ref{Lemma2} and then, starting from $v_n$, the sequence $w_n \in {\mathcal U}_{\infty}$ 
in accordance with Lemma \ref{Lemma3}, which makes $$\|L_K[w_n]\|_1\ge \|L_{K}[u_n]\|_1-\frac{1}{n}.$$ It is now immediate to see that $\|L_K[w_n]\|_1$ also tends  to $\|K\|_{\infty,1}$. This completes the proof. 
\medskip

%


\subsection{The unstable case}

In presence of instability, we can now consider, in place of ${\mathcal L}_1$, the space ${\mathcal L}$ containing all the Lebesgue-measurable functions from ${\mathbb R}_+ \setminus \mathcal{I}$ to ${\mathbb R}$, i.e. 
the space of Lebesgue measurable and real-valued  functions 
defined a.e. over ${\mathbb R}_+$.
Abusing notation,
we still use 
$$
\|y\|_1:=\int_0^{+\infty} \ |y(t)|dt, \ \forall y\in{\mathcal L}
$$
noticing that $\|\cdot \|_1$ is no more a true norm since one could have $\|y\|_1=+\infty$.\\
With these facts in mind, since instability is equivalent to 
$$
\sup_{u\in{\mathcal B}_{\infty}} \ \|L_K[u]\|_1=+\infty,
$$
we now show that 
\begin{equation}\label{FinalStep}
\sup_{u\in{\mathcal B}_{\infty}} \ \|L_K[u]\|_1=+\infty \ \Rightarrow \ \sup_{u\in{\mathcal U}_{\infty}} \ \|L_K[u]\|_1=+\infty.
\end{equation}



For this purpose, consider the restriction $K_n$ of $K$ to the square ${\mathcal Q}_n=[0,n] \times [0,n]$.
More precisely, $K_n$ coincides with $K$ in ${\mathcal Q}_n$ and assume null values elsewhere.
Since $K$ is continuous (Mercer), $K_n$ is locally bounded over ${\mathcal Q}_n$ and hence, absolutely integrable.
This implies that $K_n$ is a stable kernel. 
Define
\begin{equation}\label{Seqan}
a_n:=\|K_n\|_{\infty,1}=\sup_{u \in {\mathcal B}_{\infty}} \ \|L_{K_n}[u]\|_1=\sup_{u \in {\mathcal U}_{\infty}} \ \|L_{K_n}[u]\|_1
\end{equation}
where the last equality derives from the result obtained in the first part of the proof
applied to the stable kernel $K_n$. We now show that the sequence $a_n$ is divergent reporting the first
of two instrumental lemmas useful to obtain the desired result.\\

\begin{lemma}\label{UnSlemma1}
Let $a_n$ be defined in \eqref{Seqan}.
Then, one has
\begin{equation}\label{LimSeqan}
a:=\lim_{n\rightarrow+\infty} \ a_n = +\infty.
\end{equation}
\end{lemma}
\begin{proof}
The sequence 
$a_n$ is monotonically non decreasing. To see this, one can e.g.
apply inputs with support over $[0,n]$ to kernel operators induced by $K_m$, with $m \ge n$, obtaining outputs with $1$-norms 
not less than those defined by the same inputs applied to kernel operators associated with $K_n$.
It follows that there exists 
$$
a=\lim_{n\rightarrow+\infty} \ a_n \ge a_n, \ \forall n.
$$
For the sake of contradiction, assume $a<+\infty$. Since $K$ is unstable, there exists $v\in{\mathcal B}_{\infty}$ such that $\|L_K[v]\|_1\ge a+1$ (possibly $\|L_K[v]\|_1=+\infty${\color{blue})}. Now consider the restrictions $v_n$ of $v$ over $[0,n]$ with null values outside this interval. 
Let us introduce the sets 
$$
{\mathcal R}_n( p ):=\left\{ (x,y): 0 \le x \le n+p, \ 0 \le y \le n \right\}
$$
and
$$
{\mathcal R}_n:=\{ \ (x,y): 0 \le x <+\infty, \ 0 \le y \le n \ \}
$$
with $K_n(p)$ and $K_n(+\infty)$ to denote the corresponding ``kernels''
(we call these objects kernels improperly since they are not even symmetric).
One has
$$
L_{K_n(p)}[v_n]=L_{K_{n+p}}[v_n] \ \Rightarrow \ \|L_{K_n(p)}[v_n]\|_1 \le a_{n+p} \le a, \ \forall p
$$
where the equality holds since $v_n(t)=0$ for $t \in (n,n+p]$ while the inequality follows from the fact that all the $a_n$ 
are upper bounded by $a$. Hence, 
$$
\|L_{K_n(p)}[v_n]\|_1=\int_0^{n+p} \ |y_n(t)| dt \le a, \ \forall p \ \Rightarrow \ \int_0^{+\infty} \ |y_n(t)| dt \le a
$$
$$
\Rightarrow \ \|L_{K_n(+\infty)}[v]\|_1=\|L_{K_n(+\infty)}[v_n]\|_1\le a
$$
(note that $L_{K_n(+\infty)}[v]=L_{K_n(+\infty)}[v_n]$ since the support of $K_n$ does not extend beyond $n$). 
Now, for any $t$ such that the kernel section is absolutely integrable we define
$$
y(t):=L_K[v](t)=\int_0^{+\infty} \ K(t,\tau)v(\tau)d\tau,
$$
$$
y_n(t):=L_{K_n(+\infty)}[v](t)=\int_0^n \ K(t,\tau)v(\tau)d\tau.
$$
Then, the term
$$
|y(t)-y_n(t)|\le\int_n^{+\infty} \ |K(t,\tau)|d\tau
$$
is infinitesimal as $n$ grows to infinity, so that $[L_{K_n(+\infty)}[v](t)$ pointwise converges to $y(t)$ a.e. 
(in fact the kernel sections are non absolutely summable only if $t$ belongs to ${\mathcal I}$ 
whose Lebesgue measure has been assumed null). Summarizing, one has
$$
\|y_k(t)\|_1\le a, \ \|y(t)\|_1\ge a+1, \ y_k(t) \rightarrow y(t) \ \mbox{(pointwise a.e.)}.
$$
But if $y_k(t)$ converges pointwise a.e. to $y(t)$, a fortiori one has 
$$
y(t)=\lim_{k\rightarrow+\infty} \ \inf_{m\ge k} \ y_m(t)=\lim_{k\rightarrow+\infty} \ y_k(t) \ \mbox{a.e.}
$$
and, using the Fatou's Lemma \cite{Rudin}, one obtains
$$
\|y\|_1=\int_0^{+\infty} \ |y(t)|dt \le \lim_{k\rightarrow+\infty} \ \inf_{m\ge k} \ \int_0^{+\infty} \ |y_m(t)|dt =
$$
$$
= \lim_{k\rightarrow+\infty} \ \inf_{m\ge k} \ \|y_m(t)\|_1 \le a.
$$
This shows that $\|y\|_1\le a$ and $\|y\|_1\ge a+1$, leading to the desired contradiction if $a<+\infty$.
Hence, one must have $a=+\infty$ and this completes the proof.\\
\end{proof}

We now introduce the second lemma useful for our purposes. It provides an 
extension of the  convexity arguments contained in Lemma \ref{Lemma1} 
to handle functions not necessarily in ${\mathcal L}_1$.\\

\begin{lemma}\label{Lemma4} Consider the (possibly
not absolutely integrable) functions $y(t),w(t)$
 and define
$$
F(x):=\int_0^{+\infty} \ | y(t)+xw(t) |dt, \ x\in [-x_1, x_2], \ x_1,x_2>0.
$$
Then, at least one of $x_1,x_2$ satisfies the inequality $F(x_i)\ge F(0)$.
\end{lemma}
\begin{proof}
The possible situations are
\begin{itemize}
\item $w(t) \in {\mathcal L}_1, \ y(t) \in {\mathcal L}_1$. In this case the result follows from Lemma 
\ref{Lemma1} and basic properties of convex functions. 
\item $w(t) \in {\mathcal L}_1, \ y(t) \notin {\mathcal L}_1$. In this case, one has 
\begin{eqnarray*}
&& F(x):=\int_0^{+\infty} \ | y(t)+xw(t) |dt \\
&&  \qquad \ge \int_0^{+\infty} \ | y(t) |dt -|x|\int_0^{+\infty} \ | w(t) |dt=+\infty, \ \forall x \in {\mathbb R}
\end{eqnarray*}
and the result immediately follows. 
\item $w(t) \notin {\mathcal L}_1, \ y(t) \in {\mathcal L}_1$. In this case, 
\begin{eqnarray*}
&& F(x):=\int_0^{+\infty} \ | y(t)+xw(t) |dt \\
&&  \qquad  \ge |x|\int_0^{+\infty} \ | w(t) |dt-\int_0^{+\infty} \ | y(t) |dt=+\infty, \ \forall x \in {\mathbb R} / \{ \ 0 \}
\end{eqnarray*}
so that both $F(x_1)$ and $F(x_2)$ are larger than $F(0)$. 
\item $w(t) \notin {\mathcal L}_1, \ y(t) \notin {\mathcal L}_1$ and, in addition, there exists $x_0 \in {\mathbb R}_+$ 
such that $y(t)+x_0w(t) \in {\mathcal L}_1$. Now
\begin{eqnarray*}
&& F(x):=\int_0^{+\infty} \ | y(t)+xw(t) |dt  \\
&&  \qquad  
=\int_0^{+\infty} \ | [y(t)+x_0w(t)]  + (x-x_0)w(t) |dt
\end{eqnarray*}
so that we can consider the second case just replacing $y(t)$ with $y(t)+x_0w(t)$ and $xw(t)$ with $(x-x_0)w(t)$.
So, $F(x)$ is always equal to $+\infty$ except when $x=x_0$. So, $F(x_1)\ge F(0)$ and/or $F(x_2)\ge F(0)$. 
\item $w(t) \notin {\mathcal L}_1, \ y(t) \notin {\mathcal L}_1$ and, in addition, there does not exist $x_0 \in {\mathbb R}_+$ such that 
$y(t)+x_0w(t) \in {\mathcal L}_1$. In this case, $F(x)$ is always equal to $+\infty$.
\end{itemize}
Summarizing, we have seen that, except in the first case already treated in Lemma \ref{Lemma1}, $F(x)$ is always equal to $+\infty$ except at most when $x=x_0$. So, when $x=x_1$ and/or $x=x_2$ it assumes an infinite value. This completes the proof.\\  
\end{proof}

We are now in a position to conclude the proof.
Exploiting Lemma \ref{UnSlemma1}, in particular
\eqref{Seqan} and \eqref{LimSeqan}, we can 
find a sequence $w_n\in{\mathcal B}_{\infty}$ with bounded support over $[0,n]$
satisfying
$$
\|y_n\|_1:=\|L_{K}[w_n]\|_1 \ge a_n-\epsilon, \ |w_n(t)|=1, \ \forall t: \ 0\le t \le n.
$$
In fact, given any $\epsilon>0$, the 1-norm just of $y_n(t)$ restricted to $[0,n]$ would not be smaller than $a_n-\epsilon$ in view of the definition of supremum applied to the kernel operator associated with $K_n$. 
Thus, one has 
$$
\lim_{n \rightarrow +\infty}  \|L_{K}[w_n]\|_1 = +\infty. 
$$
Now, we extend $w_n$ for $t>n$ as follows 
$$
v_n(x):=w_n+x\psi_{(n,+\infty)} 
$$
and note that $v_n(x) \in {\mathcal B}_{\infty}$ for any $x \in [-1,1]$ and $v_n(x) \in {\mathcal U}_{\infty}$ for $x=\pm 1$. 
Now, we can exploit Lemma \ref{Lemma4}, 
interpreting $v_n(x)$ 
as $y(t)+xw(t)$. One obtains $$\|L_{K}[v_n](x_n)\|_1\ge \|L_K[w_n]\|_1 \ \ \text{if either} \ \ x_n=+1 \ \ \text{or} \ \ x_n=-1.$$
Then, choosing $$s_n:=v_n(x_n),$$ with $x_n$ depending on $n$ and chosen in order to satisfy the above inequality for any $n$, one has $$\|L_{K}[s_n]\|_1\ge a_n-\epsilon.$$
This permits to conclude that 
$$
\lim_{n \rightarrow +\infty}  \|L_{K}[s_n]\|_1 = +\infty \ \ \text{with all the} \ \ s_n\in{\mathcal U}_{\infty},
$$
hence completing the proof (note that $\|L_{K}[s_n]\|_1$ could be already equal to $+\infty$, 
making not even necessary the building of the entire sequence).

\section{The discrete-time case}\label{Sec3}

In discrete-time, the kernel $K$ can be interpreted as an infinite-dimensional matrix. 
Using notation of ordinary algebra also to handle objects of infinite dimension, the kernel operator 
$L_K[u]$ becomes $Ku$ where $u$ is an infinite-dimensional column vector. 
It is then immediate to see that all the arguments developed
in the previous section to prove Theorem \ref{MainTh} still hold.
One has just to consider the corresponding discrete-time versions
of $\mathcal{L}_{1},\mathcal{L}_{\infty},\mathcal{B}_{\infty}$ and ${\mathcal U}_{\infty}$. 
Furthermore, when discussing
the unstable case, the space of functions ${\mathcal L}$ with the Lebesgue measure 
is replaced by the measure space containing all the sequences
over $\mathbb{N}$ equipped with the counting measure. 
Then, all the arguments contained in the proof 
hold true with integrals that become sums. 
One concludes that 
even in discrete-time the sup of the kernel operator norm computed 
over ${\mathcal B}_{\infty}$ 
coincides with that over ${\mathcal U}_{\infty}$, i.e.
$$
\sup_{u \in  {\mathcal B}_{\infty}} \|Ku\|_1 = \sup_{u \in  {\mathcal U}_{\infty}} \|Ku\|_1.
$$

\section{Conclusions}

The result obtained in this paper 
gives additional insights on the relationship
between stability of a single linear dynamic systems
and stability of an ensemble of systems embedded in a 
positive definite kernel. 
Stability of a RKHS containing time-invariant linear systems can now be assessed using the same 
functions needed to investigate BIBO stability of any single element in the RKHS. 
We envision that the availability of this new stability
test (over a smaller subset of test functions) could  
facilitate the development of new theory regarding 
the use of RKHSs in system identification and control theory.


\end{document}